\documentclass{class}
\usepackage{graphicx,amssymb,amsmath}

\title{Square Formation by Asynchronous Oblivious Robots}

\author{Marcello Mamino\thanks{Institut f\"ur Algebra, TU Dresden, 01062 Dresden, Germany, {\tt marcello.mamino@tu-dresden.de}}\and Giovanni Viglietta\thanks{University of Ottawa, Canada, {\tt gvigliet@uottawa.ca}}}

\index{Mamino, Marcello}
\index{Viglietta, Giovanni}

\begin{document}
\thispagestyle{empty}
\maketitle

\begin{abstract}
A fundamental problem in Distributed Computing is the Pattern Formation problem, where some independent mobile entities, called robots, have to rearrange themselves in such a way as to form a given figure from every possible (non-degenerate) initial configuration.

In the present paper, we consider robots that operate in the Euclidean plane and are dimensionless, anonymous, oblivious, silent, asynchronous, disoriented, non-chiral, and non-rigid. For this very elementary type of robots, the feasibility of the Pattern Formation problem has been settled, either in the positive or in the negative, for every possible pattern, except for one case: the Square Formation problem by a team of four robots.

Here we solve this last case by giving a Square Formation algorithm and proving its correctness. Our contribution represents the concluding chapter in a long thread of research. Our results imply that in the context of the Pattern Formation problem for mobile robots, features such as synchronicity, chirality, and rigidity are computationally irrelevant.
\end{abstract}

\section{Introduction and Background}\label{s:1}

Consider a finite set of independent computational entities, called \emph{robots}, that live and operate in the Euclidean plane and are capable of observing each other's positions and moving to other locations, through so-called \emph{Look-Compute-Move cycles}. A fundamental motion planning question in Distributed Computing is which \emph{patterns} can be formed by such robots, regardless of their initial positions. This is known as the \emph{Pattern Formation problem}, and has been extensively studied under different robot models (see the monography~\cite{FloPS12}).

In this paper we focus on a very weak type of robots, which are modeled as geometric points (\emph{dimensionless}), are all indistinguishable from each other (\emph{anonymous}), and execute the same deterministric algorithm. Moreover, they retain no memory of past events and observations (\emph{oblivious}), they cannot communicate explicitly (\emph{silent}), they have no common notion of time (\emph{asynchronous}), of a North direction (\emph{disoriented}), of a clockwise direction (\emph{non-chiral}), and they may be unpredictably stopped during each cycle before reaching their intended destination (\emph{non-rigid movements}). This robot model is often called \emph{ASYNCH}.

Note that if $n$ such robots initially form a regular $n$-gon and their local coordinate systems are oriented symmetrically, then they all have the same ``view'' of the world. Now, if they are all activated synchronously, they are bound to make symmetric moves forever, implying that they will always form a regular $n$-gon, or perhaps collide in the center. As the pattern has to be formed from every possible initial configuration, the Pattern Formation problem is unsolvable if the pattern is not a regular polygon or a point. Clearly, the existence of a general algorithm that will indeed make the robots form a regular polygon or a point from any initial configuration is not obvious and has been the object of intensive research by several authors.

In the case of a point, the Pattern Formation problem has been eventually settled in~\cite{CieFPS12}, where an algorithm is presented that always makes $n\neq 2$ ASYNCH robots gather in a point (if $n=2$, the problem is unsolvable).

In the case of a regular polygon, there is a long history of algorithms that solve the Pattern Formation problem under increasingly weaker robot models. We start from the semi-synchronous model, \emph{SSYNCH}, in which we assume the existence of a global ``clock'' that discretizes time. In each time unit, some robots (chosen by an external ``adversary'') perform a complete Look-Compute-Move cycle synchronously, while the other robots remain inactive. In~\cite{SuzY99} it is shown that SSYNCH robots can always form a regular polygon, provided that they have the ability to remember their past observations (hence they are not oblivious). In~\cite{DefS08} the obliviousness of the robots is restored, but the algorithm proposed only makes the robots \emph{converge} to a regular polygon, perhaps without ever forming one. These results were improved in~\cite{DieP07a}, where it is shown how $n\neq 4$ SSYNCH robots, without additional requirements, can form a regular polygon. The case of a square, $n=4$, was solved separately in~\cite{DieP08} with an ad-hoc algorithm.

For ASYNCH robots, a simple solution was given in~\cite{FloPSW08}, under the assumption that the local coordinate systems of all robots have the same orientation. This result was improved in~\cite{FujYKY12}, where it is only assumed that the local coordinate systems are all right-handed (\emph{chirality}), but may be rotated arbitrarily. In~\cite{FloPSV14}, an algorithm is given for $n\neq 4$ ASYNCH robots with no assumptions on their local coordinate systems, but allowing them to move along circular arcs, as well as straight line segments. A solution for $n\neq 4$ ASYNCH robots with no extra assumptions was finally given in~\cite{FloPSV15}. The general algorithm lets only a few robots move at a time, so that the others will provide a stable ``reference frame'' for them. The case $n=4$ is left unsolved in~\cite{FloPSV15}, essentially because four robots are too few to implement this strategy, yet enough to make ad-hoc solutions elusive.

In the following sections, we formalize the \emph{Square Formation} problem for $n=4$ ASYNCH robots, we give an algorithm for it, and we prove its correctness, thus completing the characterization of the patterns that are formable by ASYNCH robots from every initial configuration. Since the proof of non-formability of asymmetric patterns that we outlined above holds even for fully synchronous robots with chirality and rigid movements (i.e., movements that cannot be unpredictably stopped by an adversary), all these features turn out to be computationally irrelevant with respect to the Pattern Formation problem.

\section{Model Specification}\label{s:2}

Let $\mathcal R=\{r_1,r_2,r_3,r_4\}$ be a set of \emph{robots}, each of which is thought of as a computational entity occupying a point in the plane $\mathbb R^2$ and having its own local Cartesian coordinate system. Each robot's coordinate system is always centered at the robot's location, and different robots' coordinate systems may have different orientation, handedness, and unit of length.

Each robot cyclically goes through three phases: \emph{Look}, \emph{Compute}, and \emph{Move}. In a Look phase, the robot takes an instantaneous ``snapshot'' of all the robots' locations and it expresses them as points in $\mathbb R^2$ within its own local coordinate system. In the next Compute phase, these four points are fed, in any order, to an \emph{algorithm} $\mathcal A$, which outputs a \emph{destination} point $p$, again expressed in the robot's coordinate system. The algorithm $\mathcal A$ is the same for all robots, and it can only compute algebraic functions of the input points (for our purposes, we will only need arithmetic functions and square roots). In the next Move phase, the robot moves toward $p$ along a straight line. Note that, even though the robots are indistinguishable, each robot can identify itself in the snapshots it takes, because it is always located at $(0,0)$.

When a Move phase ends, the Look phase of the next cycle starts. We may assume that the Look and Compute phases of a robot are executed together and instantaneously at each cycle, but the Move phase's duration may vary, although it must be finite. The duration of each Move phase of each cycle of each robot is decided arbitrarily by an external ``adversary'', called \emph{scheduler}. As a consequence, a robot may perform a Look while some other robots are in the middle of a movement, and there is no way to tell it from the snapshot. The scheduler also arbitrarily sets the speed of each robot at each moment of each Move phase; the velocity vector must always be directed toward the current destination point, or be the null vector. In particular, a robot may actually start moving a long time after the last Look-Compute phase, when the snapshot it has taken and the destination it has computed are already ``obsolete''. 

The scheduler can also decide to end a robot's Move phase before it reaches its intended destination. The only constraint is that it cannot do so before the robot has moved by at least $\delta$ during that phase, where $\delta$ is a fixed positive distance (in absolute units), not known to the robots. This is to guarantee that if a robot keeps computing the same destination point (in absolute coordinates), it reaches it in finitely many cycles.

An initial configuration of the robots is \emph{non-degenerate} if no two robots are located in the same point, and no robot is moving (formally, each robot's initial destination point coincides with the robot's initial location, and all robots are in a Move phase initially).

The Square Formation problem asks for a specific algorithm $\mathcal A$, whose input is a quadruplet of points and the output is a single destination point, such that, if the four robots of $\mathcal R$ execute $\mathcal A$ in all their Compute phases, starting from any non-degenerate initial configuration, and regardless of the scheduler's choices and of the value of $\delta$, they always end up forming a square in finite time. Once a square is formed, the robots have to maintain their positions forever. (Recall that the input quadruplet always contains $(0,0)$ as the executing robot's location, and the value of $\delta$ cannot be accessed by $\mathcal A$.)

\section{Preliminary Constructions and Definitions}\label{s:3}

\begin{figure}[!b]
\centering
\includegraphics[scale=1]{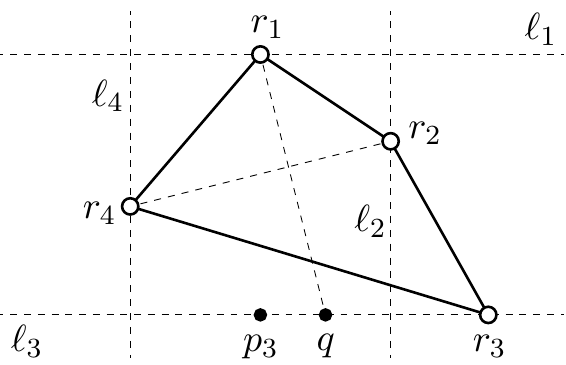}
\caption{Constructing the guidelines and the targets}
\end{figure}

The following geometric construction will be useful in our Square Formation algorithm. Let $r_1r_2r_3r_4$ be a strictly convex quadrilateral whose diagonals $r_1r_3$ and $r_2r_4$ are not orthogonal. Let $q$ be the unique point such that $r_1q=r_2r_4$, the lines $r_1q$ and $r_2r_4$ are orthogonal, and the ray emanating from $r_1$ and passing through $q$ intersects the line $r_2r_4$. Let $\ell_3$ be the line through $r_3$ and $q$, and let $\ell_1$ be the line through $r_1$ parallel to $\ell_3$. (Since $r_1r_3$ and $r_2r_4$ are not orthogonal, $\ell_3$ is well defined and is distinct from $\ell_1$.) Let $\ell_2$ be the line through $r_2$ orthogonal to $\ell_1$, and let $\ell_4$ be the line through $r_4$ parallel to $\ell_2$. By construction, these four lines intersect at four points that are vertices of a square $Q$. In turn, the midpoints of the edges of $Q$ form a second square $Q'$, called the \emph{target square}.

Ideally, if the robots of $\mathcal R$ are located at $r_1$, $r_2$, $r_3$, $r_4$ (with abuse of notation, we identify each robot with its location), our Square Formation algorithm will attempt to make the team move to the vertices of $Q'$. More precisely, the \emph{target} of $r_i$ is the vertex $p_i$ of $Q'$ that lies on $\ell_i$. The line $\ell_i$ is called the \emph{guideline} of $r_i$, and the segment $r_ip_i$ is the \emph{pathway} of $r_i$. If $r_i$ lies in the interior of an edge of $Q$, then $r_i$ is said to be \emph{internal}; otherwise, $r_i$ is \emph{external}. If two robots have parallel guidelines, they are said to be \emph{opposite} to each other.

In Section~\ref{s:5} we will prove that if we permute the labels of the vertices of the above quadrilateral arbitrarily, as long as the indices follow a clockwise or a counterclockwise order, and we repeat the same construction, the resulting target square will be the same. Hence, if several robots compute the above construction at the same time within their local coordinate systems, they necessarily obtain the same target square. Note also that the target square remains unaltered as the robots move along their pathways, as long as the quadrilateral stays convex, and its diagonals stay non-orthogonal.

Let $c$ be the center of the target square. The ``signed distance'' between $r_i$ and its target can be computed as $(r_i-c)\times(r_i'-c)$, where $\times$ denotes the cross product in $\mathbb R^2$, defined as $(x_1,y_1)\times (x_2,y_2)=x_1y_2-x_2y_1$. If this number is $0$, then $r_i$ is said to be \emph{finished}. If the product of the signed distances of two robots is not negative (respectively, not positive), the two robots are said to be \emph{concordant} (respectively, \emph{discordant}). Intuitively, they are concordant if they move around $c$ in the same ``direction'' (i.e., clockwise or counterclockwise) as they go toward their targets.

Let the guidelines of two discordant robots $r_i$ and $r_j$ intersect in a point $g$. If $p_i$ lies on the segment $r_ig$ and $p_j$ lies on the segment $r_jg$, then $r_i$ and $r_j$ are said to be \emph{convergent}; otherwise, they are \emph{divergent} (if both $r_i$ and $r_j$ are finished, they are both convergent and divergent).

If the pathway of $r_i$ intersects the segment $r_jr_k$ in $v$, then $r_i$ is said to be \emph{blocked} at $v$. If the pathway of $r_i$ intersects an extension of the segment $r_jr_k$ in $v$, then $r_i$ is said to be \emph{hindered} at $v$ (see Figure~\ref{f:2hinder}).

Let us give one last definition. A \emph{thin hexagon} is a hexagon $H=h_1h_2h_3h_4h_5h_6$ such that $h_1h_2=h_3h_4=h_4h_5=h_6h_1=h_1h_4/4$, the angles at $h_1$ and $h_4$ are $50^\circ$, and all other angles are equal. $h_1$ and $h_4$ are the \emph{extremes} of $H$, and the segment $h_1h_4$ is the \emph{main diagonal}. The other four vertices are called \emph{beacons}, and the midpoint of two adjacent beacons is a \emph{haven}.

\begin{figure}[!ht]
\centering
\includegraphics[scale=1]{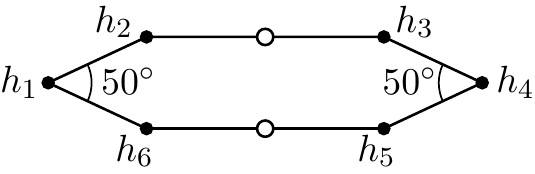}
\caption{Thin hexagon (empty dots denote havens)}
\end{figure}

\section{The Square Formation Algorithm}\label{s:4}

The Square Formation algorithm will identify the current configuration's \emph{class} among the ones listed below, and it will execute a different procedure based on the class. If the configuration belongs to several classes, the relevant one is the one that appears first in the list.

So, when a new class is defined, it is assumed that none of the definitions of the previous classes are satisfied. In particular, after quadrilaterals with orthogonal diagonals and non-convex quadrilaterals have been ruled out, the target square described in Section~\ref{s:3} will be well defined, as well as each robot's guideline, etc.

Again, we identify each robot $r_i$ with its position, and a class' definition is fulfilled when there is a permutation of the indices that satisfies the corresponding condition.

We will use expressions of the form, ``robot $r_i$ moves to point $d_i$; robot $r_j$ moves to point $d_j$'' without specifying which robot is actually running the algorithm, as if there was a global coordinator overseeing the execution. However, since the robots are anonymous and independent, we will always move symmetric robots in symmetric ways, so as to comply with the specification of the robot model given in Section~\ref{s:2}.

Borrowing from~\cite{FloPSV14,FloPSV15}, we will make use of \emph{cautious moves}, whose purpose is, roughly speaking, to prevent situations in which a robot is still in the middle of a movement when a configuration class change occurs. This is done by identifying a finite number of \emph{critical points} and making each moving robot stop at the first critical point it encounters. In this section, we will only generically say, ``robot $r_i$ cautiously moves toward $d_i$'', leaving out the tricky details of how critical points are chosen and maintained. A complete analysis of every case will be carried out in Section~\ref{s:5}, along with the proof of correctness of the Square Formation algorithm.\\
\\
\textbf{Configuration 1: Orthogonal}\\
\textbf{Definition.} The segments $r_1r_3$ and $r_2r_4$ are orthogonal and intersect in a point $c$ (possibly an endpoint).\\
\textbf{Execution.} Each of $r_1$, $r_2$, $r_3$, $r_4$ moves away from $c$, to a point at distance $\max\{r_1c,r_2c,r_3c,r_4c\}$ from it.\\
\\
\textbf{Configuration 2: Thin Hexagon}\\
\textbf{Definition.} The thin hexagon $H$ with main diagonal $r_1r_2$ contains also $r_3$ and $r_4$, but not on two adjacent beacons.\\
\textbf{Execution.} If both $r_3$ and $r_4$ are on the main diagonal, they move orthogonally to it, remaining within $H$. If $r_3$ is on the main diagonal and $r_4$ is not, $r_3$ moves to the opposite side of the main diagonal, orthogonally to it, remaining within $H$. If $r_3$ and $r_4$ are on different sides of the main diagonal, they move to the haven on their respective side; the closest moves first while the other one waits (in case of a tie, they both move). If $r_3$ and $r_4$ are on the same side of the main diagonal, they move to the two adjacent beacons on that side, minimizing the total distance traveled.\\
\\
\textbf{Configuration 3: Non-Convex}\\
\textbf{Definition.} The triangle $r_1r_2r_3$ contains $r_4$.\\
\textbf{Execution.} $r_4$ moves to the foot of an altitude of $r_1r_2r_3$ that lies in the interior of an edge of $r_1r_2r_3$.\\
\\
\textbf{Configuration 4: Pinwheel}\\
\textbf{Definition.} $r_1$, $r_2$, $r_3$, $r_4$ are all concordant.\\
\textbf{Execution.} Let $r_1$ be opposite to $r_3$ and assume without loss of generality that $r_1r_3<r_2r_4$ (if $r_1r_3=r_2r_4$, the configuration is Orthogonal). Suppose that $r_1$ and $r_3$ are both finished. If the thin hexagon $H$ with main diagonal $r_2r_4$ contains $r_1$ (and not $r_3$), assume without loss of generality that $r_2r_1<r_2r_3$, and let $r_2$ cautiously move toward its target. If neither $r_1$ nor $r_3$ is in $H$, both $r_2$ and $r_4$ cautiously move toward their targets. Suppose now that $r_1$ and $r_3$ are not both finished. if $r_1$ is blocked at $v$, it cautiously moves toward $v$. If neither $r_1$ nor $r_3$ is blocked and $r_1$ is hindered at $v$, $r_1$ cautiously moves toward $v$. If neither $r_1$ nor $r_3$ is blocked or hindered, $r_1$ and $r_3$ cautiously move toward their targets. (These rules are exhaustive due to Lemma~\ref{l:pinwheel1} below.)\\
\\
\textbf{Configuration 5: Scissors}\\
\textbf{Definition.} $r_1$, $r_2$ are divergent; $r_3$, $r_4$ are divergent.\\
\textbf{Execution.} If exactly one robot is external, it cautiously moves toward its target. If all the internal robots are finished, all the external robots cautiously move toward their targets. In all other cases, all the internal robots cautiously move toward their targets.\\
\\
\textbf{Configuration 6: Flowing}\\
\textbf{Definition.} $r_1$, $r_2$ are divergent; $r_3$, $r_4$ are convergent.\\
\textbf{Execution.} If two opposite robots are finished, the non-finished one that is closest to its target cautiously moves toward it (there cannot be a tie, or the configuration would be Orthogonal). Otherwise, if exactly one of $r_1$ and $r_2$ is finished, the robot opposite to it cautiously moves toward its target. Otherwise, both $r_3$ and $r_4$ cautiously move toward their targets.\\
\\
\textbf{Configuration 7: One Discordant}\\
\textbf{Definition.} $r_1$, $r_2$, $r_3$ are concordant; $r_4$ is discordant.\\
\textbf{Execution.} $r_4$ cautiously moves toward its target.

\section{Correctness of the Algorithm}\label{s:5}

As noted in Section~\ref{s:3}, the target square is well defined, no matter how each robot computes it.

\begin{lemma}\label{l:target1}
Given a strictly convex quadrilateral with non-orthogonal diagonals, regardless of how labels $r_1$, $r_2$, $r_3$, $r_4$ are assigned to its vertices following a clockwise or a counterclockwise order, the construction in Section~\ref{s:3} yields the same guidelines and target square.
\end{lemma}
\begin{proof}
The construction does not change if we invert the labels of $r_2$ and $r_4$, hence we may assume that the indices are arranged in clockwise order. Now it is enough to prove that the construction does not change if we shift the labels clockwise by one position. So, let $q'$ be the unique point such that $r_2q'=r_1r_3$, the lines $r_2q'$ and $r_1r_3$ are orthogonal, and the ray emanating from $r_2$ and passing through $q'$ intersects the line $r_1r_3$. By construction, the triangle $r_2r_4q'$ is a copy of $r_1r_3q$ rotated by $90^\circ$, which means that the line $r_4q'$ is orthogonal to $\ell_3$, and hence coincident with $\ell_4$. It follows that all the new guidelines are the same as in the original construction, and so is the target square.
\end{proof}

\begin{lemma}\label{l:target2}
Given a strictly convex quadrilateral with non-orthogonal diagonals, if its vertices are labeled $r_1$, $r_2$, $r_3$, $r_4$ in clockwise order, then their targets $p_1$, $p_2$, $p_3$, $p_4$ also appear in clockwise order, and vice versa.
\end{lemma}
\begin{proof}
It suffices to prove that if $p_1$, $p_2$, $p_3$ are in clockwise order, then so are $r_1$, $r_2$, $r_3$. Referring to Figure~\ref{f:order}, if $r_1$, $r_2$, $r_3$ are in counterclockwise order, then $r_2$ is located to the right of the line $r_2r_3$. Therefore, $q$ must be to the left of $\ell_2$, contradicting the fact that $q$ must lie on $\ell_4$, which in turn is located to the right of $\ell_2$.
\begin{figure}[!ht]
\centering
\includegraphics[scale=1]{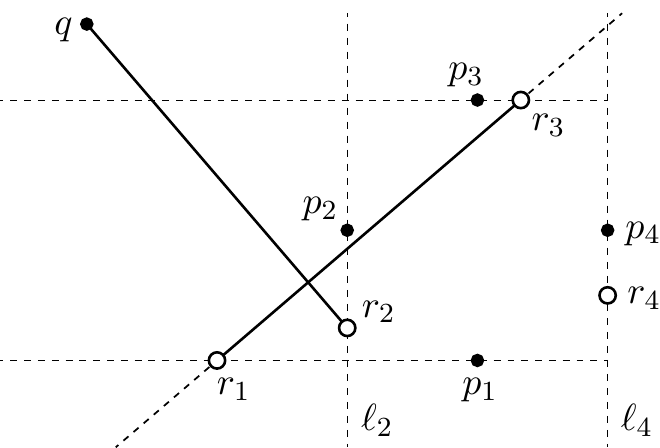}
\caption{$r_1$, $r_2$, $r_3$ cannot be in counterclockwise order}
\label{f:order}
\end{figure}
\end{proof}

The following is an easy consequence of Lemma~\ref{l:target2}.

\begin{cor}\label{c:external}
No two robots have intersecting pathways. If two robots are external, they are either concordant or convergent. If three robots are external, they are concordant.\hfill$\square$
\end{cor}

\begin{lemma}\label{l:pinwheel1}
In a Pinwheel configuration, two opposite robots cannot be both blocked. Moreover, if $r_1r_3<r_2r_4$, robots $r_1$ and $r_3$ cannot be both hindered.
\end{lemma}
\begin{proof}
Referring to Figure~\ref{f:2hinder}, if $r_2$ is blocked, it means that the segment $r_1r_3$ crosses $\ell_2$ above $p_2$. Hence $r_1r_3$ must cross $\ell_4$ further above, and therefore $r_4$ cannot be blocked, because its pathway lies below $p_4$.

Suppose now that $r_1$ and $r_3$ are both hindered, as in Figure~\ref{f:2hinder}. Because the line $r_1r_4$ intersects the pathway of $r_3$, it is immediate to see that $r_1$ must be external and $r_4$ must be internal. Symmetrically, $r_3$ must be external and $r_2$ must be internal. This means that $r_1r_3>r_2r_4$, contradicting our assumption.
\begin{figure}[!ht]
\centering
\includegraphics[scale=1]{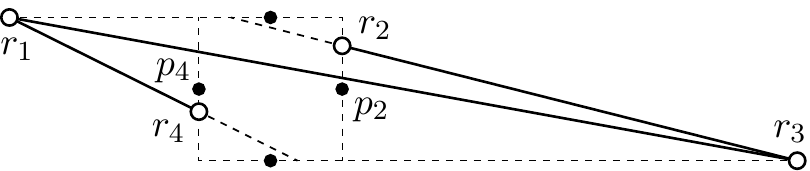}
\caption{$r_2$ is blocked; $r_1$ and $r_3$ are hindered}
\label{f:2hinder}
\end{figure}
\end{proof}

\begin{lemma}\label{l:pinwheel2}
In a Pinwheel configuration where $r_1$ and $r_3$ are opposite, both are finished, and neither of them is in the thin hexagon $H$ with main diagonal $r_2r_4$, there is a choice of critical points that lets $r_2$ and $r_4$ reach their targets in finite time as they perform a cautious move.
\end{lemma}
\begin{proof}
Our critical points will prevent $r_2$ and $r_4$ from forming an Orthogonal, Thin Hexagon, or Non-Convex configuration before reaching their targets.

An Orthogonal configuration cannot be formed before $r_2$ and $r_4$ reach their targets, because $r_1r_3$ is parallel to their guidelines.

A Thin Hexagon configuration cannot be formed, either. Note that $r_1$ and $r_3$ must be on opposite sides of $H$, due to Lemma~\ref{l:target2}. As Figure~\ref{f:nohex2} suggests, it is straightforward to verify that if $r_1$ and $r_3$ are both outside $H$, the sum of the distances of $r_1$ and $r_3$ from the main diagonal is greater than the height of $H$ (this is true because each short edge of a thin hexagon is a quarter of its main diagonal). As $r_2$ and $r_4$ move toward their targets, the main diagonal becomes shorter and the sum of the distances of $r_1$ and $r_3$ from the main diagonal grows. In particular, this sum is greater than the new height of $H$. It follows that, no matter how $r_2$ and $r_4$ move, a Thin Hexagon configuration will never be formed.

\begin{figure}[!ht]
\centering
\includegraphics[scale=1]{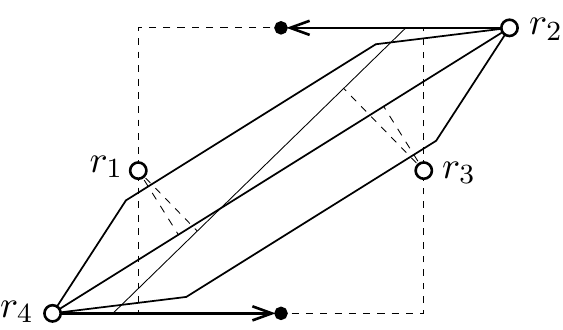}
\caption{As $r_2$ and $r_4$ move, no Thin Hexagon is formed}
\label{f:nohex2}
\end{figure}

Observe that both $r_2$ and $r_4$ could be hindered, say at $v$ and $v'$, respectively (as $r_1$ and $r_3$ in Figure~\ref{f:2hinder}). It suffices to set their critical points halfway to $v$ and $v'$ to guarantee that a Non-Convex configuration will never be formed. We call the segments $r_2v$ and $r_4v'$ \emph{safe zones}. It is easy ot see that, as $r_2$ and $r_4$ move (and recompute a new $v$ and $v'$ pair), their safe zones get longer, giving them even more leeway, until they are not hindered any more, and can safely reach their targets.
\end{proof}

\begin{lemma}\label{l:scissors}
In a Thin Hexagon configuration that is also a Scissors one, both extremes must be external.
\end{lemma}
\begin{proof}
By Corollary~\ref{c:external}, at most two robots can be external. Clearly, each external robot must necessarily be an extreme of the thin hexagon. Suppose for a contradiction that at most one robot is external. Let $r_1$ be either an external robot or an extreme of the thin hexagon (in case there are no external robots). The other extreme must be the farthest robot, hence either $r_3$ or $r_4$ (assuming that $r_1$ and $r_2$ are divergent). Now it is straightforward to verify that both angles $\angle r_1r_3r_4$ and $\angle r_1r_4r_3$ are greater than $\arctan(1/2)>25^\circ$, and hence the thin hexagon cannot contain both $r_3$ and $r_4$.
\end{proof}

\begin{lemma}\label{l:nohex}
If a configuration is not Thin Hexagon, $r_3$ and $r_4$ are convergent, and $r_1$ and $r_2$ do not move while $r_3$ and $r_4$ move toward their targets, the configuration never becomes Thin Hexagon.
\end{lemma}
\begin{proof}
By Lemma~\ref{l:target2}, $r_1$ and $r_2$ are on the same side of the line $r_3r_4$, as in Figure~\ref{f:nohex1}. Let $H$ be the thin hexagon with main diagonal $r_3r_4$. By assumption, either $r_1$ and $r_2$ occupy two adjacent beacons of $H$ or one of them, say $r_1$, is not in $H$. In both cases, as soon as $H$ starts moving together with $r_3$ and $r_4$, $r_1$ is guaranteed to remain strictly outside of $H$. Indeed, the position of $H$ at each time can be obtained from its initial position by a composition of two types of transformations: shrinkage about an extreme and rotation about an extreme in the direction opposite to $r_1$. Both these operations cause $r_1$ to stay out of $H$ if it is already out and to get out of $H$ if it is initially on a beacon.
\begin{figure}[!hb]
\centering
\includegraphics[scale=1]{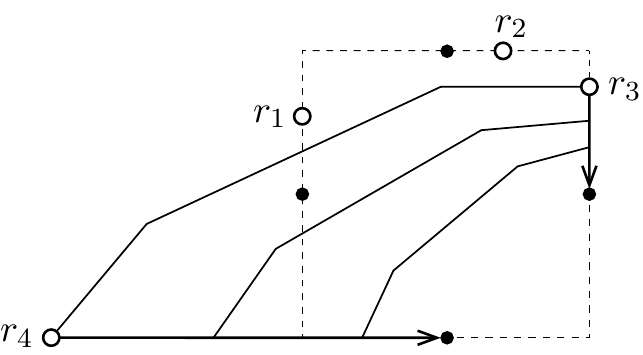}
\caption{As $r_3$ and $r_4$ move, no Thin Hexagon is formed}
\label{f:nohex1}
\end{figure}
\end{proof}

\begin{theorem}
The algorithm of Section~\ref{s:4} solves the Square Formation problem in the ASYNCH model.
\end{theorem}
\begin{proof}
We prove that the rules given in Section~\ref{s:4} are well defined, exhaustive, and make the four robots form a square in a finite amount of time. The general idea is that each configuration can only remain in the same class, or transition to a class with lower index, with one exception: a Thin Hexagon can become a Scissors configuration, and this case will be examined separately. We also have to verify that no robot is moving when the team's behavior changes, in order to prevent inconsistencies arising from robots believing to be in different classes, due to asynchronicity. The cautious move technique serves this purpose, but we have to show that suitable critical points exist. When only one robot is tasked with moving, it is sufficient to identify the first location on its path that may cause other robots to start moving. So, no discussion will be needed in this case.

Verifying the above for Configurations~1--3 is trivial, so let us assume that the robots' initial configuration is none of those, and in particular that a target square is well defined (cf.~Lemma~\ref{l:target1}). Now, as the robots move toward their targets, the target square remains unaltered, and collisions are impossible due to Corollary~\ref{c:external}.

Let the configuration be in the Pinwheel class. Lemma~\ref{l:pinwheel1} implies that the rules are unambiguous and a robot always moves, eventually forming a lower-index configuration or reaching a target. If $r_1r_3<r_2r_4$, the main diagonal of a possible thin hexagon $H$ must be $r_2r_4$, so it is easy for $r_1$ and $r_3$ to stop as soon as they reach the boundary of $H$, since $H$ remains still as they move. When $r_1$ and $r_3$ are finished, the cautious move of $r_2$ and $r_4$ succeeds due to Lemma~\ref{l:pinwheel2}.

Let the configuration be in the Scissors class. By Lemma~\ref{l:scissors}, if there are fewer than two external robots, no Thin Hexagon can ever be formed, and if there are two external robots (not more, by Corollary~\ref{c:external}), there is only one candidate thin hexagon $H$ having these two robots as extremes. If the internal robots move, they can set their critical points on the boundary of $H$. If the external robots move, they must be convergent (or they would be concordant by Corollary~\ref{c:external}, and the configuration would be Pinwheel), hence no Thin Hexagon can be formed, due to Lemma~\ref{l:nohex}. Also, no moving robot can be blocked or hindered at any point. When two opposite robots are finished, the configuration may transition to Pinwheel while other robots are still moving, but this is fine because it does not cause a change in behavior.

Let the configuration be in the Flowing class. If $r_3$ and $r_4$ move together, no critical points originating from thin hexagons have to be set, due to Lemma~\ref{l:nohex}, and the other critical points are easy to spot. If two opposite robots are finished, moving only the non-finished robot closest to its target prevents the formation of an Orthogonal configuration. A transition to Pinwheel or Scissors may occur, but only when no robots are moving.

If the configuration class is none of the above, there must be a unique discordant robot, which turns the configuration into a Pinwheel one upon reaching its target.

As we mentioned, the only anomaly in this process is the transition from a Thin Hexagon to a Scissors configuration, which occurs when two robots reach adjacent beacons. According to the rules for the Scissors case, the two robots on the beacons (which are internal) move to their targets. Since they move away from the main diagonal, they can form no Thin Hexagon. Afterwards, the two external robots (which are convergent) move to their targets without forming a Thin Hexagon, by Lemma~\ref{l:nohex}. No Orthogonal or Non-Convex configurations can be formed during these procedures, either.
\end{proof}

\small
\bibliographystyle{abbrv}

\begin{thebibliography}{99}

\bibitem{CieFPS12}
M.~Cieliebak, P.~Flocchini, G.~Prencipe, and N.~Santoro.
\newblock Distributed computing by mobile robots: gathering.
\newblock \emph{SIAM Journal on Computing}, 41(4):829--879, 2012.

\bibitem{DefS08}
X.~D\'efago and S.~Souissi.
\newblock Non-uniform circle formation algorithm for oblivious mobile robots with convergence toward uniformity.
\newblock \emph{Theoretical Computer Science}, 396(1--3):97--112, 2008.

\bibitem{DieP07a}
Y.~Dieudonn\'e and F.~Petit.
\newblock Swing words to make circle formation quiescent.
\newblock \emph{14th International Colloquium on Structural Information and Communication Complexity (SIROCCO)}, 166--179, 2007.

\bibitem{DieP08}
Y.~Dieudonn\'e and F.~Petit.
\newblock Squaring the circle with weak mobile robots.
\newblock \emph{19th International Symposium on Algorithms and Computation (ISAAC)}, 354--365, 2008.

\bibitem{FloPS12}
P.~Flocchini, G.~Prencipe, and N.~Santoro.
\newblock \emph{Distributed computing by oblivious mobile robots}.
\newblock Synthesis Lectures on Distributed Computing Theory,
\newblock Morgan \& Claypool, 2012.

\bibitem{FloPSV14}
P.~Flocchini, G.~Prencipe, N.~Santoro, and G.~Viglietta.
\newblock Distributed computing by mobile robots: solving the uniform circle formation problem.
\newblock \emph{18th International Conference on Principles of Distributed Systems (OPODIS)}, 217--232, 2014.

\bibitem{FloPSV15}
P.~Flocchini, G.~Prencipe, N.~Santoro, and G.~Viglietta.
\newblock Distributed computing by mobile robots: solving the uniform circle formation problem.
\newblock arXiv:1407.5917 [cs.DC], 2015.

\bibitem{FloPSW08}
P.~Flocchini, G.~Prencipe, N.~Santoro, and P.~Widmayer.
\newblock Arbitrary pattern formation by asynchronous oblivious robots.
\newblock {\em Theoretical Computer Science}, 407(1--3):412--447, 2008.

\bibitem{FujYKY12}
N.~Fujinaga, Y.~Yamauchi, S.~Kijima, and M.~Yamashita.
\newblock Asynchronous pattern formation by anonymous oblivious mobile robots.
\newblock \emph{26th International Symposium on Distributed Computing (DISC)}, 312--325, 2012.

\bibitem{SuzY99}
I.~Suzuki and M.~Yamashita.
\newblock Distributed anonymous mobile robots: formation of geometric patterns.
\newblock \emph{SIAM Journal on Computing}, 28(4):1347--1363, 1999.

\end{thebibliography}

\end{document}